\documentclass[11pt]{amsart} 
\usepackage{amscd,amssymb,amsxtra}
\usepackage{mathrsfs}  
\DeclareSymbolFontAlphabet{\mathrsfs}{rsfs}
\usepackage[mathscr]{eucal} 
\usepackage{comment}
\usepackage{color}
\usepackage{enumitem} 
\usepackage[colorlinks,backref=page,citecolor=refkey]{hyperref}
\usepackage{aliascnt}

\usepackage{marginnote}

\makeatletter
\let\@secnumfont\bfseries
\def\section{\@startsection{section}{1}%
  \z@{4\linespacing\@plus\linespacing}{\linespacing}%
  {\bfseries\centering}}
\def\introsection{\@startsection{section}{1}%
  \z@{3\linespacing\@plus\linespacing}{\linespacing}%
  {\bfseries\centering}}
\def\subsection{\@startsection{subsection}{2}%
   \z@{1.25\linespacing\@plus.7\linespacing}{.5\linespacing}%
   {\normalfont\bfseries}}
\def\subsectionsinline{\def\subsection{\@startsection{subsection}{2}%
  \z@{1\linespacing\@plus.7\linespacing}{-.5em}%
  {\normalfont\bfseries}}}

\makeatother

\numberwithin{equation}{section}

\newcommand{\mynewtheorem}[2]{
  \newaliascnt{#1}{equation}
  \newtheorem{#1}[#1]{#2}
  \aliascntresetthe{#1}
  \expandafter\def\csname #1autorefname\endcsname{#2}
}

\theoremstyle{definition}
\mynewtheorem{definition}{Definition}
\mynewtheorem{example}{Example}
\mynewtheorem{problem}{\color{blue}Problem}
\mynewtheorem{probsec}{\color{blue}Problem}
\mynewtheorem{exercise}{Exercise}
\mynewtheorem{question}{\color{blue}Question}
\mynewtheorem{project}{\color{blue}Project}
\mynewtheorem{construction}{Construction}
\mynewtheorem{task}{\color{blue}Task}
\mynewtheorem{otask}{\color{green}Obsolete Task}
\mynewtheorem{notation}{Notation}

\newtheorem*{definition*}{Definition}
\newtheorem*{example*}{Example}
\newtheorem*{problem*}{\color{blue}Problem}
\newtheorem*{probsec*}{\color{blue}Problem}
\newtheorem*{exercise*}{Exercise}
\newtheorem*{question*}{\color{blue}Question}
\newtheorem*{project*}{\color{blue}Project}
\newtheorem*{construction*}{Construction}
\newtheorem*{notation*}{Notation}

\theoremstyle{remark}
\mynewtheorem{note}{Note}
\mynewtheorem{remark}{Remark}
\mynewtheorem{data}{Data}

\newtheorem*{note*}{Note}
\newtheorem*{remark*}{Remark}
\newtheorem*{data*}{Data}

\theoremstyle{plain}
\mynewtheorem{theorem}{Theorem}
\mynewtheorem{corollary}{Corollary}
\mynewtheorem{lemma}{Lemma}
\mynewtheorem{proposition}{Proposition}
\mynewtheorem{conjecture}{Conjecture}
\mynewtheorem{claim}{Claim}
\mynewtheorem{proposal}{Proposal}
\mynewtheorem{conclusion}{Conclusion}
\mynewtheorem{hypothesis}{Hypothesis}
\mynewtheorem{assumption}{Assumption}

\newtheorem*{theorem*}{Theorem}
\newtheorem*{corollary*}{Corollary}
\newtheorem*{lemma*}{Lemma}
\newtheorem*{proposition*}{Proposition}
\newtheorem*{conjecture*}{Conjecture}
\newtheorem*{claim*}{Claim}
\newtheorem*{proposal*}{Proposal}
\newtheorem*{conclusion*}{Conclusion}
\newtheorem*{hypothesis*}{Hypothesis}
\newtheorem*{assumption*}{Assumption}

\newenvironment{proof*}[1][\proofname]{
  \begin{proof}[#1]}{  
\end{proof}}

\definecolor{refkey}{rgb}{0,.6,.4}

\renewcommand{\:}{\colon}

\newcommand{\CC}{{\mathbb C}}

\DeclareMathOperator{\Diff}{Diff}

\DeclareMathOperator{\Hom}{Hom}
\DeclareMathOperator{\id}{id}

\newcommand{\PP}{{\mathbb P}}

\newcommand{\QQ}{{\mathbb Q}}

\newcommand{\RR}{{\mathbb R}}

\DeclareMathOperator{\tr}{tr}

\newcommand{\ZZ}{{\mathbb Z}}

\newcommand{\chiup}{\raise.5ex\hbox{$\chi$}}

\newcommand{\inv}{^{-1}}
\DeclareRobustCommand{\mstrut}{^{\vphantom{1*\prime y\vee M}}}

\DeclareMathOperator{\rank}{rank}

\newcommand{\temsquare}{\raise3.5pt\hbox{\boxed{ }}}

\newcommand{\zmod}[1]{\ZZ/#1\ZZ}

\DeclareFontFamily{U}{mathx}{}
\DeclareFontShape{U}{mathx}{m}{n}{<-> mathx10}{}
\DeclareSymbolFont{mathx}{U}{mathx}{m}{n}
\DeclareMathAccent{\widehat}{0}{mathx}{"70}
\DeclareMathAccent{\widecheck}{0}{mathx}{"71}
 
\DeclareMathSymbol{\bigtimes}{1}{mathx}{"91}

\DeclareMathOperator{\SO}{SO}

\DeclareMathOperator{\U}{U}
\DeclareMathOperator{\SU}{SU}

\DeclareMathOperator{\PU}{PU}

\usepackage[all,2cell]{xy}
\usepackage{graphicx}\usepackage{epsf}

\definecolor{refkey}{rgb}{0,.8,.2}\definecolor{labelkey}{rgb}{1,0,0} 

\newcounter{alphthm}
\theoremstyle{plain}
\newtheorem{mainthm}[alphthm]{Theorem}
\newtheorem*{ratcri*}{Rationality Criterion}

\DeclareMathOperator{\Det}{Det}

\newcommand{\BG}{BG}
\newcommand{\BSO}{B\!\SO}
\newcommand{\BUo}{B\!\U(1)}
\newcommand{\Cx}{\CC^{\times}}
\newcommand{\OsE}{\Omega _{\sE}}
\newcommand{\QZ}{\QQ/\ZZ}
\newcommand{\RZ}{\RR/\ZZ}

\newcommand{\SH}{\mathfrak{u}}
\newcommand{\aIR}{\alpha \mstrut _{\textnormal{IR}}}
\newcommand{\aUV}{\alpha \mstrut _{\textnormal{UV}}}
\newcommand{\ath}{\alpha \mstrut _{\theta }}

\newcommand{\bth}{\beta \mstrut _{\theta }}
\newcommand{\fpP}{f _{(p,P)}}
\newcommand{\fpjPj}{f_{(p_j,P_j)}}
\newcommand{\gk}{\gamma \mstrut _{\kappa }}
\newcommand{\iR}{\sqrt{-1}\RR}
\newcommand{\rA}{I}
\newcommand{\rF}{\mathrsfs{F}}
\newcommand{\sE}{\mathscr{E}}
\newcommand{\sG}{\mathscr{G}}
\newcommand{\sH}{\mathscr{H}}
\newcommand{\sX}{\mathscr{X}}
\newcommand{\sY}{\mathscr{Y}}
\newcommand{\tS}{\widetilde{S}}
\newcommand{\tc}{\tilde{c}}
\newcommand{\tpi}{2\pi \sqrt{-1}}

\begin{document}

\abovedisplayskip18pt plus4.5pt minus9pt
\belowdisplayskip \abovedisplayskip
\abovedisplayshortskip0pt plus4.5pt
\belowdisplayshortskip10.5pt plus4.5pt minus6pt
\baselineskip=15 truept
\marginparwidth=55pt

\makeatletter
\renewcommand{\tocsection}[3]{%
  \indentlabel{\@ifempty{#2}{\hskip1.5em}{\ignorespaces#1 #2.\;\;}}#3}
\renewcommand{\tocsubsection}[3]{%
  \indentlabel{\@ifempty{#2}{\hskip 2.5em}{\hskip 2.5em\ignorespaces#1%
    #2.\;\;}}#3} 
\renewcommand{\tocsubsubsection}[3]{%
  \indentlabel{\@ifempty{#2}{\hskip 5.5em}{\hskip 5.5em\ignorespaces#1%
    #2.\;\;}}#3} 
\def\@makefnmark{%
  \leavevmode
  \raise.9ex\hbox{\fontsize\sf@size\z@\normalfont\tiny\@thefnmark}} 
\def\multfoot{\textsuperscript{\tiny\color{red},}}
\def\footref#1{$\textsuperscript{\tiny\ref{#1}}$}

\makeatother

\newcommand{\bmuu}{\mbox{$\raisebox{-.07em}{\rotatebox{9.9}
  {\tiny {\bf /}
  }}\hspace{-0.53em}\mu\hspace{-0.88em}\raisebox{-0.98ex}{\scalebox{2} 
  {$\color{white}\phantom{.}$}}\hspace{-0.416em}\raisebox{+0.88ex}
  {$\color{white}\phantom{.}$}\hspace{0.46em}$}} 
\newcommand{\bmut}{\bmu 2}
\newcommand{\bmu}[1]{\bmuu _{#1}}

\setcounter{tocdepth}{2}



 \title[Anomalies, gaps, and torsion]{Gapped theories have torsion anomalies} 

  \author[C.~C\'ordova]{Clay C\'{o}rdova}
  \thanks{CC is supported by the Simons Foundation Award 888980 as part of the
Simons Collaboration on Global Categorical Symmetries, as well as by the Sloan
Foundation and by DOE Early Career Grant 5-29073.}
  \address{Department of Physics, Kadanoff Center for Theoretical Physics \&
Enrico Fermi Institute, University of Chicago}
  \email{clayc@uchicago.edu}
  
  \author[D. S. Freed]{Daniel S.~Freed}
 \thanks{DSF is supported by the National Science Foundation under Grant Number
DMS-2005286 and by the Simons Foundation Award 888988 as part of the Simons
Collaboration on Global Categorical Symmetries.}
 \address{Harvard University \\ Department of Mathematics \\ Science Center
Room 325 \\ 1 Oxford Street \\ Cambridge, MA 02138}
 \email{dafr@math.harvard.edu}

  \author[C. Teleman]{Constantin Teleman} 
  \thanks{CT is supported by the Simons Foundation Award 824143 as part of the
Simons Collaboration on Global Categorical Symmetries.}
  \address{Department of Mathematics \\ University of California \\ 970 Evans
 Hall \#3840 \\ Berkeley, CA 94720-3840}  
  \email{teleman@math.berkeley.edu}

  \thanks{We thank Thomas Dumitrescu and Ryan Thorngren for informative and
insightful discussions.}

 \date{August 27, 2024}
 \begin{abstract} 
 We prove special cases of a general conjecture: If an invertible field theory
admits a projectively topological boundary theory, then it has finite order in
the abelian group of invertible field theories.  One can substitute `gapped'
for `projectively topological'.  Our proofs use evaluations of a field theory
in parametrized families.
 \end{abstract}
\maketitle

Fix $\theta \in \RZ$ and let $\ath $ denote the fully local invertible
2-dimensional field theory whose partition function on a closed 2-manifold~$X$
is
  \begin{equation}\label{eq:1}
     \ath(X,\Theta ) = \exp\left( -\theta \int_{X}\Omega (\Theta )  \right) .
  \end{equation}
The background fields on~$X$ are an orientation and a
$\U(1)$-connection~$\Theta $ with curvature~$\Omega (\Theta )$.

  \begin{mainthm}[]\label{thm:1}
 If $\ath$~admits a topological boundary theory~$F$, then $\theta $~is a
rational number.
  \end{mainthm}

  \begin{remark}[]\label{thm:2}
\ 
 \begin{enumerate}[label=\textnormal{(\roman*)}]

 \item If $F$~exists, then we say that $F$~is an anomalous topological field
theory with anomaly~$\ath$.

 \item More precisely, the theory $F$~is \emph{projectively} topological.

 \item In general, if $\rF$~is an $n$-dimensional quantum field theory with
anomaly an invertible $(n+1)$-dimensional field theory~$\aUV$, then under
renormalization group flow the anomaly~$\aUV $ deforms to an invertible
theory~$\aIR$; for the theory~$\ath$ in~\eqref{eq:1}, this means that the
constant~$\theta $ can ``run''.  (The deformation class of~$\aUV $, which is a
discrete invariant, does not change.)  If a field theory is \emph{gapped}, then
one expects at long distance that it flows to a projectively topological field
theory: $\aIR $~admits a projectively topological boundary theory.  Hence, if
one knows the flow of the isomorphism class of~$\aUV$ under renormalization,
Theorem~\ref{thm:1} provides an obstruction to the existence of gapped
anomalous quantum field theories with specified anomaly.
See~\cite{CDFKS1,CDFKS2} for a class of theories in which one knows this flow
of the anomaly theory.

 \item The partition function~\eqref{eq:1} only depends on the principal
$\U(1)$-bundle $P\to X$ that underlies the connection~$\Theta $:
  \begin{equation}\label{eq:2}
     \ath(X,P) = \exp\left( \tpi\,\theta \,c_1(P)[X] \right) . 
  \end{equation}
In a restricted sense---see~(vi) below---$\ath$~is a topological field theory
that depends only on a principal $\U(1)$-bundle.  There is an
``integer-valued'' invertible field theory with partition function $c_1(P)[X]$
from which $\ath$~is obtained by exponentiation.  The real characteristic
class~$\theta c_1\in H^2(\BUo;\RZ)$ determines the theory~$\ath$.  The
condition in the theorem is that this class lie in the image of
$H^2(\BUo;\QZ)$, or equivalently that the isomorphism class of~$\ath$ have
finite order in the abelian group of 2-dimensional invertible field theories
(with given background fields).

 \item Theorem~\ref{thm:1}, and its cousins below, are closely related to: The
twisting of a finite rank twisted vector bundle has finite order in the abelian
group of twistings of $K$-theory.  (This is used in the proof of
\autoref{thm:12} below.)  The ``period-index problem'' asks for the minimal
rank of a twisted vector bundle that realizes a given finite order twisting,
and this problem can be posed for general Brauer groups; see~\cite{AW} for the
case of twisted $K$-theory.

 \item Our proof relies on evaluation of~$\ath$ and~$F$ in parametrized
families.  One motivation for this note is to exhibit the utility of
evaluation in families, even in topological field theory, an idea that can be
traced back to the 1990s, for example in~\cite{KM}.  A very small selection of
recent references is~\cite{CFLS1,CFLS2,HKT}.

 \item The invertible theory~$\ath$ is topological when evaluated on a single
manifold, as explained in~(iii) above, but it is {not} topological when
evaluated in families.  See \cite[Appendix~B]{FN} for more discussion of this
distinction.

 \item Theorem~\ref{thm:1} answers a very special case of a general problem:
When does an invertible field theory admit a projectively topological boundary
theory?  This question has received substantial attention in the physics
literature.  The earliest investigations involve perturbative chiral anomalies,
whose inflow theories are suitable generalizations of Chern-Simons terms.  In
this context it is known that the quantized level appearing in the inflow term
controls the coefficient of a local correlation function with power law decay
and hence enforces that the theory is gapless~\cite{HIJLMSS,FSBY,CG}. More
recently, this problem has been generalized to include variations where the
symmetry type is finite and one imposes the additional hypothesis that the
symmetry is not spontaneously broken.  This is enforced mathematically by
requiring that the spectrum of the theory on certain spheres is suitably
trivial.  Again one finds that often anomalies can obstruct topological
boundary theories, a phenomenon known as \emph{symmetry enforced gaplessness}.
See~\cite{WNMXS,CO1,CO2} for examples.

 \end{enumerate}
  \end{remark}

Another example of interest is the 4-dimensional invertible theory~$\bth$,
defined on oriented Riemannian manifolds, whose partition function on a closed
oriented 4-manifold is
  \begin{equation}\label{eq:3}
     \bth(X) = \exp\left( \tpi\,\theta \,p_1(TX)[X] \right) . 
  \end{equation}

  \begin{mainthm}[]\label{thm:3}
  If $\bth$~admits a projectively topological boundary theory~$F$, then $\theta
$~is a rational number.
  \end{mainthm}

\noindent
 In fact, this is the example relevant to~\cite{CDFKS1,CDFKS2}. 

  \begin{example}[]\label{thm:5}
 A modular tensor category~$A$ is an invertible object in the 4-category of
braided tensor categories, and when supplemented with $\SO_4$-invariance data
it determines via the cobordism hypothesis a 4-dimensional invertible
topological field theory~$\beta $ of oriented manifolds.  The regular module
of~$A$---equipped with $\SO_3$-invariance data---determines a topological
boundary theory~$F$, which is a Reshetikhin-Turaev theory.  The chiral central
charge $c\in \QQ/8\ZZ$ is computed from the algebraic data in~$A$, and part of
the data is a lift to~$\tc\in \QQ/24\ZZ$.  Then $\beta =\bth$ with $\theta
=\tc/24\pmod1$.  This point of view on Reshetikhin-Turaev theories is due to
Walker~\cite{Wa}; see~\cite{Ha} for recent developments using the cobordism
hypothesis.
  \end{example}

Here is a general result that we can prove using our
techniques.\footnote{Analogous results hold for characteristic numbers in
generalized cohomology theories.}  Let $G$~be a Lie group, let $n$~be a
positive integer, and consider characteristic classes of smooth oriented
$(n+1)$-manifolds~$X$ equipped with a principal $G$-bundle $P\to X$.  A
characteristic class
  \begin{equation}\label{eq:14}
     \kappa \in H^{n+1}(\BSO\times \BG;\RZ) 
  \end{equation}
defines an invertible $(n+1)$-dimensional field theory~$\gk $ with partition
function
  \begin{equation}\label{eq:4}
     \gk (X,P )=\exp\left( \tpi\,\kappa (TX,P)[X] \right).
  \end{equation}
Let $p\:X\to S$ be an oriented fiber bundle with $\dim X=n+1$, and let $P\to
X$ be a principal $G$-bundle.  Suppose in addition that $S$~is oriented.
Then $(p,P)$~determines a linear functional
  \begin{equation}\label{eq:15}
     \begin{aligned} \fpP\:H^{n+1}(\BSO\times \BG;\RR)&\longrightarrow \RR \\
      \lambda &\longmapsto p_*\bigl[\lambda\bigl(T(X/S),P \bigr)
      \bigr]\,[S]\end{aligned} 
  \end{equation}

  \begin{ratcri*}[]\label{thm:4}
 Assume there exists a finite set~$\bigl\{(p_j,P _j)\bigr\}_{j\in J}$ in which
\textnormal{(i)}~$p_i\:X_i\to S_i$~is an oriented fiber bundle, $\dim X_i=n+1$,
$\dim S_i=2$, $S_i$~is oriented; and \textnormal{(ii)}~the functionals $\bigl\{
\fpjPj \bigr\}_{j\in J} $ form a basis of $H^{n+1}(\BSO\times \BG;\RR)^*$.
Then if $\gk $~ admits a projectively topological boundary theory~$F$, the
class~$\kappa $ lies in the image of $H^{n+1}(\BSO\times \BG;\QZ)$.
  \end{ratcri*}

\noindent
 The existence of the basis $\bigl\{ (p_j,P _j) \bigr\}_{j\in J} $ is surely
not necessary, but rather is a limitation of the technique we employ in this
note.  It is reminiscent of conditions in the early literature on anomalies,
for example~\cite{K,FV}, which are eliminated when stronger locality is
imposed.  Here we only assume that $\gk $~is an $(n-1,n,n+1)$-theory and
$F$~is an $(n-1,n)$-theory.

  \subsection*{Quantum mechanics in a family}

Let $F$~be an $n$-dimensional \emph{topological} field theory as in~\cite{A1}.
(We leave the background fields---orientation, spin structure,
etc.---implicit.)  If $Y$~is a closed $(n-1)$-manifold, then $F(Y)$~is a
complex vector space.  Furthermore, diffeomorphisms of~$Y$ act on the vector
space~$F(Y)$ \emph{topologically}: the action factors through $\pi _0\Diff(Y)$.
This implies that if $\rho \:\sY\to S$ is a smooth fiber bundle whose fibers
are closed $(n-1)$-manifolds, then the values of~$F$ on the fibers form a
\emph{flat} complex vector bundle $F(\rho )=F(\sY/S )\to S$: a smooth complex
vector bundle with a flat covariant derivative.  For later reference, we
memorialize the conclusion. 

  \begin{proposition}[]\label{thm:14}
 Let $F$~be an $n$-dimensional topological field theory, and suppose $\rho
\:\sY\to S$ is a smooth fiber bundle of closed $(n-1)$-dimensional manifolds.
Then $F$~evaluates on~$\rho $ to a vector bundle $F(\rho )\to S$ with a flat
covariant derivative. 
  \end{proposition}

  \begin{remark}[]\label{thm:8}
 If $\sX\to S$ is a fiber bundle of bordisms from incoming boundaries
$\rho _0\:\sY_0\to S$ to outgoing boundaries $\rho _1\:\sY_1\to S$, then 
  \begin{equation}\label{eq:7}
     F(\sX/S)\in H^0\bigl(S;\Hom\bigl(F(\rho_0),F(\rho_1) \bigr) \bigr) 
  \end{equation}
is a \emph{flat} section of $\Hom\bigl(F(\rho _0),F(\rho _1) \bigr)\to S$.
In \emph{cohomological field theory} \eqref{eq:7}~is promoted to a cohomology
class of arbitrary (mixed) degree; see~\cite{EY,DVV,LZ,G,KM}, for example.
There is a cochain refinement due to Segal~\cite{S}.  
  \end{remark}

In the general (nontopological) case, a family of quantum mechanical systems
parametrized by a smooth manifold~$S$ is formulated on a continuous fiber
bundle $\sH\to S$ of Hilbert spaces; see \cite[Appendix~D]{FM} for a careful
treatment.  One should postulate a dense subbundle $\sH_\infty \to S$ that is
smooth and carries a covariant derivative that need not be flat.  In
\autoref{thm:14} we encounter a family of \emph{topological} quantum mechanical
systems inside a single topological field theory, and $\sH\to S$ has finite
rank and a flat structure.

The pure states in a quantum mechanical system form a \emph{projective} space,
not a \emph{linear} space.  The projectivity can be considered to be anomaly of
the theory.  The anomaly of an $n$-dimensional quantum field theory includes
the projectivity of the state space of a closed $(n-1)$-manifold.  We model the
projectivity in a family of quantum mechanical systems parametrized by a smooth
manifold~$S$ as a gerbe with connection over~$S$.  Its isomorphism class lies
in the differential cohomology group~$\widecheck{H}^3(S)$, and there are
several models possible for a cocycle or geometric representative.  (Any
geometric representative of a degree~3 differential cohomology class may be
considered to be a `gerbe with connection', though some authors may prefer a
more restricted usage of the term.)  We turn now to a \v{C}ech model, both for
the differential cohomology cocycle and for a twisted vector bundle with
covariant derivative.  (See~\cite[\S2]{F2} for an alternative model.)

  \subsection*{Twisting by a gerbe}

Let $S$~be a smooth manifold, and suppose $\{U_i \}_{i \in \rA}$ is
an open cover.  The \v{C}ech data for a gerbe with connection $\sG\to S$,
implicitly assumed \emph{topologically} trivialized on each~$U_i $, is: 
  \begin{equation}\label{eq:8}
     \begin{aligned} &B_i \in \Omega ^2(U_i ;\iR), \\ &L_{i
      j }\longrightarrow U_i \cap U_j , \\ &\varphi \mstrut
      _{i j k }\:\underline{\CC}\xrightarrow{\;\;\cong
      \;\;}L\mstrut _{j k }\otimes L\inv _{i k }\otimes
      L\mstrut _{i j }\qquad \textnormal{on $U_i \cap U_j \cap
      U_k $},\end{aligned} 
  \end{equation}
where $L_{i j }\to U_i \cap U_j $ is a hermitian line bundle with unitary
covariant derivative, $\underline{\CC}\to U_i\cap U_j\cap U_k$ is the product
line bundle with constant fiber~$\CC$, and $\varphi \mstrut _{ijk}$~is a
\emph{flat} isomorphism of line bundles.  Let $\omega _{i j }\in \Omega
^2(U_i \cap U_j;\iR)$ be the curvature of $L_{i j }\to U_i \cap U_j $.  Then
the conditions imposed on the data~\eqref{eq:8} are:
  \begin{equation}\label{eq:9}
     \begin{aligned} B_j -B_i &=\omega _{i j }\qquad
      &&\textnormal{on $U_i \cap U_j $}, \\ \varphi \mstrut _{j
      k \ell }\otimes \varphi \inv _{i k \ell }\otimes
      \varphi \mstrut _{i j \ell }\otimes \varphi \inv _{i
      j k } &= \id\qquad &&\textnormal{on $U_i \cap U_j
      \cap U_k \cap U_\ell $}.\end{aligned} 
  \end{equation}
The differentials~$dB_i $ patch to a global closed 3-form $H\in \Omega
^3(S;\iR)$, the \emph{curvature} of $\sG\to S$.  We refer to~\cite{H} for
more on gerbes with connection, and in particular for the definition
of a trivialization.

  \begin{definition}[]\label{thm:9}
 The gerbe with connection $\sG\to S$ is \emph{flat} if $dB_i =0$ for all $i
\in \rA$.
  \end{definition}

\noindent
 A flat gerbe determines a cohomology class $\mu \in H^2(S;\RZ)$, the de Rham
class of $\sqrt{-1}B/2\pi $; equivalent flat gerbes give the same cohomology
class.  The values of~$\mu $ on cohomology classes in~$H_2(S)$ are the
\emph{monodromies} of the flat gerbe.

The state spaces of an anomalous family of quantum mechanical systems with
anomaly $\sG\to S$ form a \emph{$\sG$-twisted hermitian bundle $\sE\to S$ with
covariant derivative} on~$S$.  We sketch this notion for the finite rank case
of interest here.  The \v{C}ech data for $\sE\to S$ is:
  \begin{equation}\label{eq:10}
     \begin{aligned} &(E_i ,\nabla _i )\longrightarrow U_i \\
      &\psi _{i j }\:L_{i j }\otimes E_j
      \longrightarrow E_i \qquad \textnormal{on $U_i \cap U_j$
      },\end{aligned} 
  \end{equation}
where $E_i \to U_i $ is a (finite rank) hermitian vector bundle with unitary
covariant derivative~$\nabla _i $, and $\psi _{i j }$~is an isomorphism of
vector bundles.  Let the curvature of~$\nabla _i $ be $F_i \in \Omega
^2\bigl(U_i ,\SH(E_i )\bigr)$, a skew-hermitian valued 2-form.  The
conditions imposed on the data~\eqref{eq:10} are:
  \begin{equation}\label{eq:11}
     \begin{aligned} F_i +B_i &= \psi _{i j }(F_j +B_j )\qquad &&\textnormal{on
     $U_i \cap U_j 
      $}, \\ \psi \mstrut
      _{i k } \circ \varphi\inv  _{i j k }\circ \psi \inv _{j k } \circ \psi
       _{i j }\inv  &= \id\mstrut _{E_i }\qquad &&\textnormal{on $U_i \cap U_j \cap
      U_k $}.\end{aligned} 
  \end{equation}
In the first expression, the imaginary 2-forms $B_i ,B_j ,\omega _{i j }$ act
on~$E_i ,E_j $ as scalar skew-hermitian endomorphisms.  The bundles of Lie
algebras $\mathfrak{u}(E_i)\to U_i$ patch to a global bundle of Lie algebras
using~$\psi _{ij}$, since $\mathfrak{u}(E_j)$ and $\mathfrak{u}(L_{ij}\otimes
E_j)$ are canonically isomorphic.  The first equation in~\eqref{eq:11} tells
that $F_i +B_i $ patch to a global 2-form~$\OsE$ with values in this bundle of
Lie algebras.  We call $\OsE$ the \emph{curvature} of the $\sG$-twisted vector
bundle $\sE\to S$.

  \begin{remark}[]\label{thm:10}
 $\sE\to S$ is not a global vector bundle.  Rather, the projectivizations $\PP
E_i \to U_i $ patch to a global projective bundle using~$\psi _{ij}$.  Also,
  \begin{equation}\label{eq:16}
  \begin{aligned}
   \rank\sE&\in  H^0(S; \ZZ) \\ 
   c_1(\sE)&\in H^2\bigl(S;\zmod{(\rank\sE)}\bigr)
  \end{aligned}
  \end{equation}
are well-defined.  For the Chern class, observe that if $E\to S$ is a rank~$r$
vector bundle and $L\to S$ is a line bundle, then $\Det(L\otimes E)\cong
L^{\otimes r}\otimes \Det E$.  Identify $\zmod r\cong \bmu r$ via $k\mapsto
e^{2\pi ik/r}$, where $\bmu r\subset \Cx$ is the abelian group of
$r^{\textnormal{th}}$ roots of unity.  Then the Chern class is a connecting
homomorphism in the long exact sequence of cohomology derived from the central
group extension $\bmu r\to \SU_r\to \PU_r$; it is the obstruction to lifting a
projective bundle to a vector bundle with trivialized determinant line bundle.
  \end{remark} 

  \begin{remark}[]\label{thm:13}
 The unitary structure does not enter our arguments.  We include it since
physical theories are unitary. 
  \end{remark}
 
  \begin{definition}[]\label{thm:11}
 The twisted bundle $\sE\to S$ is \emph{projectively flat} if $\OsE=0$. 
  \end{definition}

\noindent
 Applying the differential~$d$ to the trace of $F_i+B_i$, we learn that if
$\sE\to S$ is projectively flat, then the gerbe $\sG\to S$ is flat: $dB_i=0$.
(Recall that $d(\tr F_i)=0$ by the Bianchi identity.)

If $\sG\to S$ is topologically trivial, then we can and do choose the singleton
open cover $\rA=\{S\}$, in which case the only data~\eqref{eq:8} is the global
2-form $B\in \Omega ^2(S;\iR)$.  Suppose in addition that $\sG\to S$ is flat,
so that $dB=0$.  A $\sG$-twisted hermitian bundle with covariant derivative is
then given by~\eqref{eq:10} as a hermitian vector bundle $E\to S$ with unitary
covariant derivative, and it is projectively flat if its curvature $F\in \Omega
^2\bigl(S;\SH(E) \bigr)$ satisfies
  \begin{equation}\label{eq:12}
     F+B=0. 
  \end{equation}

  \begin{proposition}[]\label{thm:12}
 Let $\sG\to S$ be a flat gerbe with equivalence class $\mu \in H^2(S;\RZ)$.
Let $\sE\to S$ be a projectively flat $\sG$-twisted hermitian bundle of
rank~$r$.  Then
  \begin{equation}\label{eq:13}
     \mu =-\frac{c_1(\sE)}{r} \pmod1.
  \end{equation} 
In particular, $\mu $~lies in the image of~$H^2(S;\QZ)$. 
  \end{proposition}

  \begin{proof}
 Suppose first that $\sG\to S$ is topologically trivial, and represent $\sE\to
S$ as a global hermitian vector bundle $E\to S$, as above.  Then
multiply~\eqref{eq:12} by~$\sqrt{-1}/2\pi $ and take the trace to
deduce~\eqref{eq:13}.  In the general case, the existence of a finite rank
$\sG$-twisted bundle $\sE\to S$ implies that the cohomology class $\mu \in
H^2(S;\RZ)$ of the flat gerbe $\sG\to S$ is torsion.  To see this, apply the
determinant to~\eqref{eq:10} to conclude that $\{\Det E_i\}$ trivializes a
power of the gerbe, represented by~$\{L_{ij}^{\otimes r}\}$.  Then
$\sE^{\otimes r}\to S$ is $\sG^{\otimes r}$-twisted and the previous argument
applies.  To execute it we compute $c_1(E^{\otimes r})=r^rc_1(E)$.
  \end{proof}

  \subsection*{Proofs of the theorems}

We first explain the hypotheses of Theorem~\ref{thm:1} more precisely.  Whereas
$\ath$~is a \emph{topological} field theory when restricted to single manifolds
and bordisms, it is \emph{not} topological as a theory on smooth families of
manifolds and bordisms.  However, it is \emph{flat} (as long as $\theta $~is
fixed).  For example, a family of 0-dimensional data parametrized by a smooth
manifold~$S$ is: a finite cover $p\:\tS\to S$, a relative orientation on~$p$,
and a principal $\U(1)$-bundle $P\to\tS$ \emph{with connection}~$\Theta $.
Then $\ath(\tS/S,\Theta )\to S$ is a flat gerbe over~$S$.  A model is as
follows: use a multiplicative pushforward to construct a principal
$\U(1)$-bundle $p_*P\to S$ with connection~$p_*\Theta $, and set $B=-\theta\,
\Omega (p_*\Theta )=-\theta \,p_*\Omega (\Theta )\in \Omega ^2(S;\iR)$.  This
is \v{C}ech data~\eqref{eq:8} for a topologically trivial flat gerbe
$\ath(\tS/S,\Theta )\to S$.  Note that it depends on the connection~$\Theta $
through its curvature~$\Omega (\Theta )$, not just on the principal
$\U(1)$-bundle $P\to\tS$.

  \begin{proof}[Proof of Theorem~\ref{thm:1}]
 Let $F$~be a \emph{projectively topological} boundary theory of~$\ath$, i.e.,
an $\ath$-twisted topological quantum mechanics.  As with~$\ath$, the smooth
families on which we evaluate~$F$ carry differential data, as does the output
of~$F$.  For example, $F(\tS/S,\Theta )\to S$ is an $\ath(\tS/S,\Theta
)$-twisted hermitian bundle with a \emph{projectively flat} connection.  (This
is a relative version of \autoref{thm:14}.)  Specialize to the family of
positively oriented points $S^2\xrightarrow{\;\id\;}S^2$, orient~$S^2$, and let
$P\to S^2$ be a principal $\U(1)$-bundle of degree~1 with connection~$\Theta $.
The class of the flat gerbe $\ath(S^2/S^2,\Theta )\in H^2(S^2;\RZ)\cong \RZ$
is~$-\theta $.  Proposition~\ref{thm:12} implies that $\theta \in \QZ$.
  \end{proof}

  \begin{proof}[Proof of Theorem~\ref{thm:3}]
 Here we only\footnote{as opposed to a fully extended $(0,1,2,3,4)$-theory on
bordisms with arbitrary corners.} assume that $\bth$~is a $(2,3,4)$-theory and
the {projectively topological} boundary theory~$F$ is a $(2,3)$-theory.  Let
$p\:X\to S$ be an oriented fiber bundle over an oriented surface~$S$ with $\dim
X=4$ and $p_1\bigl(T(X/S) \bigr)[X]\neq 0$; see~\cite{A2} for an existence
proof.  As for~$\ath$, the fiber bundle~$p$ carries differential data, here a
Riemannian structure\footnote{The symbol~`$g^{X/S}$' denotes an inner product
on the relative tangent bundle $T(X/S)\to X$; we also need a horizontal
distribution on $p\:X\to S$ (that is not explicitly notated here).}~$g^{X/S}$.
Then there is a Chern-Weil 4-form~$p_1(g)$ constructed from the curvature
of~$g^{X/S}$.  The value~$\bth(X/S)$ is the topologically trivial flat gerbe
over~$S$ whose global closed 2-form is $-\theta \int_{X/S}p_1(g)$.  Proceed as
in Theorem~\ref{thm:1}, now with the additional input that the de Rham
cohomology class $p_*\bigl[p_1(X/S) \bigr]$ of $\int_{X/S}p_1(g)$ is nonzero.
  \end{proof}

  \begin{proof}[Proof of Rationality Criterion]
 Observe that the functionals~\eqref{eq:15} extend to~$\RZ$-coefficients, and
that the characteristic class $\kappa \in H^{n+1}(\BSO\times \BG;\RZ)$ lies in
the image of $H^{n+1}(\BSO\times \BG;\QZ)$ iff $\fpjPj(\kappa )\in \QZ$ for
all~$j\in J$.  This is the case if a projectively topological boundary
theory~$F$ exists: evaluate~$F$ on each~$(p_j,P_j)$ and apply the previous
arguments.
  \end{proof}

 \bigskip\bigskip
\newcommand{\etalchar}[1]{$^{#1}$}
\providecommand{\bysame}{\leavevmode\hbox to3em{\hrulefill}\thinspace}
\providecommand{\MR}{\relax\ifhmode\unskip\space\fi MR }
\providecommand{\MRhref}[2]{%
  \href{http://www.ams.org/mathscinet-getitem?mr=#1}{#2}
}
\providecommand{\href}[2]{#2}

  \end{document}